%% file: perscachearxiv.tex
\documentclass[11pt,svgnames,draft]{article}
\usepackage{authblk}

\input{ozcommands}

\usepackage{fullpage}

\begin{document}
\input{title}

\input{abstract}
\input{contentv2}

\bibliographystyle{plain}
\bibliography{dblp}

\end{document}

%% file: ozcommands.tex
\usepackage{amsthm} 
\usepackage{amsmath}
\usepackage{amssymb}
\usepackage{stmaryrd}
\usepackage{wasysym}
\usepackage{pifont}
\usepackage{lscape}
\usepackage{cite}
\usepackage{multicol}
\usepackage{caption}
\usepackage{graphicx}
\usepackage{subfig}
\usepackage{color}
\usepackage{transparent}
\usepackage{import}
\graphicspath{{figures/}}
\usepackage{pbox}
\usepackage{multirow}

\usepackage{tikz}
\usepackage{tikz-qtree}
\pdfpageattr{/Group <</S /Transparency /I true /CS /DeviceRGB>>}

\usepackage[linesnumbered,rightnl,boxed,vlined]{algorithm2e}

\usepackage{comment}

\newcommand{\Rom}[1]{\expandafter\@slowromancap\romannumeral #1@}
\makeatother

{\makeatletter
 \gdef\xxxmark{%
   \expandafter\ifx\csname @mpargs\endcsname\relax 
     \expandafter\ifx\csname @captype\endcsname\relax 
       \marginpar{xxx}
     \else
       xxx 
     \fi
   \else
     xxx 
   \fi}
 \gdef\xxx{\@ifnextchar[\xxx@lab\xxx@nolab}
 \long\gdef\xxx@lab[#1]#2{{\bf [\xxxmark #2 ---{\sc #1}]}}
 \long\gdef\xxx@nolab#1{{\bf [\xxxmark #1]}}
 \gdef\turnoffxxx{\long\gdef\xxx@lab[##1]##2{}\long\gdef\xxx@nolab##1{}}%
}

\newcommand{\cswitch}[1]{}

\newcommand{\comon}[1]{#1}
\newcommand{\comoff}[1]{}

\newcommand{\ozc}[1]{\cswitch{\comon{\textcolor{red}{\xxx[oz]{#1}}}}}

\newcounter{single}
\newtheorem{thm}[single]{Theorem}
\newtheorem{lem}[single]{Lemma}

\newcommand{\ens}[1]{\ensuremath{#1}}

\newcommand{\bigoh}[1]{\ens{\mathcal{O}(#1)}}
\newcommand{\bigom}[1]{\ens{\Omega(#1)}}
\newcommand{\thet}[1]{\ens{\Theta(#1)}}
\newcommand{\varbigoh}[1]{\ens{\mathcal{O}\left(#1\right)}}

\definecolor{nts}{rgb}{.8,.1,.8}
\definecolor{chk}{rgb}{.8,.2,.1}
\definecolor{reword}{rgb}{.9,.5,.7}

\newcommand{\nts}[1]{\textcolor{nts}{\ens{\rightarrow} #1 \ens{\leftarrow}}}
\renewcommand{\nts}[1]{ }
\newcommand{\chk}[1]{\textcolor{chk}{\ens{>~!} #1 \ens{!~<}}}
\renewcommand{\chk}[1]{ }

\newcommand{\imp}[1]{\textcolor{Red}{#1}}
\renewcommand{\imp}[1]{ }

\newcommand{\cons}[1]{\ens{C_{\mathsf{#1}}}}

\newcommand{\defcons}[2]{
	\ifbool{numsubs}{
		\edef#1{\noexpand\cons{\arabic{myc}}}
		\addtocounter{myc}{1}
	}
	{
		\newcommand{#1}{\cons{#2}}
	}
}

\newcommand{\inout}[1]{\ens{#1}} 
\renewcommand{\inout}[1]{\[#1\]}

\newcommand{\parg}[1]{\ens{(#1)}}
\newcommand{\pargg}[2]{\ens{(#1,#2)}}

\newcommand{\cods}{cache-oblivious data structure}

\newcommand{\miniexp}{\ens{\frac{1-\epsilon}{\log 3}}}
\newcommand{\miniexpin}{\ens{(1-\epsilon)/\log 3}}

\newcommand{\vers}{\ens{V}} 
\newcommand{\asize}{\ens{U}}

\newcommand{\gridwidth}{\ens{\asize}}

\newcommand{\mem}{\ens{A}}
\newcommand{\memel}[1]{\ens{\mem[#1]}}

\newcommand{\opstyle}[1]{\ens{\mathsf{#1}}}

\newcommand{\kwread}{\opstyle{Read}}
\newcommand{\readop}[1]{\kwread\ens{(#1)}}

\newcommand{\kwwrite}{\opstyle{Write}}
\newcommand{\writeop}[2]{\kwwrite\ens{(#1,#2)}}

\newcommand{\kwpread}{\opstyle{Persistent\text{-}Read}}
\newcommand{\preadop}[2]{\kwpread\ens{(#1,#2)}}

\newcommand{\kwquerytime}{\ens{t_{q}}}
\newcommand{\querytimeb}[1]{\kwquerytime\parg{#1}}
\newcommand{\querytimemb}[2]{\kwquerytime\pargg{#1}{#2}}
\newcommand{\querytimebthet}[1]{\querytimeb{\thet{#1}}}
\newcommand{\querytimembthet}[2]{\querytimemb{\thet{#1}}{\thet{#2}}}


%% file: title.tex
\author[1]{Pooya Davoodi\thanks{Research supported by NSF grant CCF-1018370 and BSF grant 2010437.}
\hspace{-.75mm}}
\author[2]{Jeremy T. Fineman\thanks{Research supported by NSF grants CCF-1314633 and CCF-1218188.}
\hspace{-.75mm}}
\author[1]{John Iacono\thanks{Research supported by NSF grant CCF-1018370.}
\hspace{-.75mm}}
\author[1]{\"Ozg\"ur \"Ozkan
\hspace{-.75mm}}
\affil[1]{New York University}
\affil[2]{Georgetown University}

\date{}

\title{Cache-Oblivious Persistence}

\maketitle

%% file: abstract.tex

\begin{abstract}
Partial persistence is a general transformation that takes a data structure and allows queries to
be executed on any past state of the structure. 
The cache-oblivious model is the leading model of a modern multi-level memory hierarchy.
We present the first general transformation for making cache-oblivious model data structures partially persistent.
\end{abstract}

%% file: contentv2.tex

\input{intro}

\section{Preliminaries}


First we define precisely what a cache-oblivious data structure is, and how its runtime is formulated.
A \cods{} manipulates a contiguous array of cells in memory by performing read and write operations on the cells of this array, which we denote by \mem{}. The array \mem{} represents the lowest level of the memory hierarchy. By the ideal cache assumption~\cite{DBLP:journals/talg/FrigoLPR12}, cache-oblivious algorithms need not (and in fact must not) manage transfers between levels of memory; they simply work with the lowest level.

\paragraph{Ephemeral primitives.}

The ephemeral data structure can be viewed as simply wishing to interact with lowest level of memory using reads and writes:

\begin{itemize}
\item
\readop i: returns \memel{i}. 
\item
\writeop ix:  sets $\memel{i}=x$.
\end{itemize}

\paragraph{Ephemeral cache-oblivious runtime and space usage.}
Given a sequence of primitive operations, we can define the runtime to execute the sequence in the cache-oblivious model. The runtime of a single primitive operation depends on the sequence of primitives executed before it, and is a function of $M$ and $B$ that returns either zero if the desired element is in a block in memory or one if it is not. This is computed as follows: for a given $B$, view $A$ as being partitioned into contiguous blocks of size $B$. Then, compute the most recent $M/B$ blocks that have had a primitive operation performed in them. If the current operation is in one of those $M/B$ blocks, its cost is zero; if it is not, its cost is one.

Space usage is simply defined as the largest memory address touched so far, which we require to be upper bounded by a linear function in the number of \kwwrite{} operations. 

\paragraph{Persistence.}

In order to support partial persistence, the notion of a version number is needed. We let $V$ denote the total number of versions so far, which is defined to be the number of \kwwrite{} operations executed so far. Let $A_v$ denote the state of memory array $A$ after the $v$th \kwwrite{} operation. Then, supporting partial persistence is simply a matter of supporting one additional primitive, in addition to \kwread{} and \kwwrite{}:

\begin{itemize}
\item
\preadop vi: returns $A_v[i]$
\end{itemize}

We wish to provide a data structure that will support persistent operations with runtime in the cache-oblivious model as close as possible to that of structures that support only ephemeral operations.

\section{Data Structure}

We view the problem in two dimensions, where an integer grid represents space-time. We say there is a point labeled $x$ at $(i,v)$ if at time $v$ a \writeop{i}{x} was executed, setting $A[i]=x$. Thus the $x$ axis represents space and the $y$ axis represents time; we assume the coordinates increase up and to the right. Let $P$ refer to the set of all points associated with all \kwwrite{} operations. In such a view $P$ is sufficient to answer all \kwread{} and \kwpread{} queries. A \readop{i} simply returns the label of the highest point in $P$ with $x$-coordinate $i$, and \preadop{v}{i} returns the label of the highest point in $P$ at or directly below $(i,v)$. We refer to the point $(i,v)$ to be the value associated with memory location $i$ at time $v$, that is, $A_v[i]$. We denote by $V$ the number of \kwwrite{} operations performed, which is the index of the most recent version.

All of the points lie in a rectangular region bounded by horizontal lines representing time zero at the bottom and the most recent version at the top, and vertical lines representing array location zero (on the left) and the maximum space usage so far (on the right). At a high level, we store a decomposition of this rectangular region into vertically-separated square regions. For each of these square-shaped regions, we use a balanced hierarchical decomposition of the square that stores the points and supports needed queries, which we call the ST-tree. As new points are only added to the top of the structure, only the top square's ST-tree needs to support insertion. As such, the non-top squares' ST-trees are stored in a more space-efficient manner and as new squares are created the old ones are compressed.

\subsection{Space-Time Trees}

We define the \textit{Space-Time Tree} or ST-tree which is the primary data structure we use to store the point set $P$ and support fast queries on $P$. 
This tree is composed of nodes, each of which has an associated space-time rectangle (which we simply call the rectangle of the node). 
See Figure~\ref{fig:space-time}.
The tree has certain properties:

\begin{itemize}

\item Each node at height $h$ in the tree (a leaf is defined to be at height 0) corresponds to a rectangle of width $2^{h}$. This implies all leaves are at the same depth.

\item Internal nodes have two or three children. An internal node's rectangle is partitioned by its children's rectangles.

\item A leaf is \emph{full} if and only if it contains a point in $P$. 
An internal node is full if and only if it has two full children. 
If an internal node has three children, one must be full.

\item Some rectangles may be three sided and open in the upward direction (future time); these are called \emph{open}, while four-sided rectangles are called \emph{closed}. All open rectangles are required to be non-full.

\end{itemize}
\input{figure1}

The above conditions imply that a node's rectangle is partitioned by the rectangles of half the width belonging to its children which we call left, right, and if there is a third one, upper. 
Each node stores:

\begin{itemize}

\item Pointers to the three children (left, right and upper) and parent.

\item The coordinates of its rectangle.

\end{itemize}
Additionally, each leaf node (height 0 and thus width 1) stores the following

\begin{itemize}

\item The at most 1 point in $P$ that intersects the rectangle of the leaf.

\item The results of a \kwpread{} corresponding to the point at the bottom of the rectangle.

\end{itemize}

\begin{lem}
\label{lem:rectlb}
If a node at height $h$ is full, then its rectangle intersects at least $2^h$ points in $P$.
\end{lem}

\begin{proof}
Follows directly from the fact that full nodes have at least two full children, full leaves intersect one point in $P$ contained in their rectangle, and the rectangles of all leaves are disjoint and at the same level.
\end{proof}

\subsection{Global data structure}
\label{subsec:globalds}

\begin{itemize}

\item 
The variable $U$ will be stored and will represent the space usage of the data structure rounded up to the next power of two.

\item The data structure will store an array of $\lceil \frac{V}{U} \rceil$ ST-trees. The last one is called the \emph{top} tree, and the others are called \emph{bottom} trees.
The tree's root's rectangles will partition the grid $[1..U]\times[1..\infty]$.
Bottom tree's roots correspond to squares of size $U$; the $j$-th bottom root corresponds to square $[1..U]\times[((j-1)U+1)..jU]$. The top tree's rectangle is the remaining three-sided rectangle $[1..U]\times[((\lceil \frac{V}{U} \rceil-1)U+1)..\infty]$.

\item An array storing a log of all $V$ \kwwrite{} operations which will be used for  rebuilding. 

\item A current memory array, call it $C$, which is of size $U$. The entry $C[i]$ contains the value of $A[i]$ at the present, and a pointer to the leaf containing the highest point in $P$ in column $i$; which also contains the value of $A[i]$ at the present.

\item A pointer $p$ to the leaf corresponding to the most recent \kwpread{}.

\end{itemize}

\subsection{How the ST-trees are stored}

The roots of all ST-trees correspond to rectangles of width $U$ and thus have height $\log U$, and structurally are subgraphs of the complete 3-ary tree of height $\log U$. 

The top tree is stored in memory in a brute force way as a complete $3$-ary tree of height $\log U$ using a biased Van Emde Boas layout \cite{DBLP:journals/talg/FrigoLPR12}\ozc{toJI: where is the biased ref?}. This layout depends on a constant $0<\epsilon<1$ and can be viewed as partitioning a tree of height $h$ into a top tree of height $\epsilon h$ and $3^{\epsilon h}$ bottom trees of height ${h(1-\epsilon)}$. Each of the nodes of these trees is then placed recursively into memory.
This will waste some space as nodes and subtrees that do not exist in the tree will have space reserved for them in memory. Thus the top ST-tree uses space $U^{\log 3}$.

There will be a level of the van Emde Boas layout  that includes the leaves and has size in the range
$3^{\log_3 B}=B$ to $3^{(1-\epsilon )\log_3 B}=B^{1-\epsilon}$ nodes.
Any path of length $k$ in an ST-tree will be stored in \bigoh{1+\frac{k}{\log B}} blocks. Additionally, we have the following lemma:

\begin{lem} 
\label{lem:binarysubgraph}
Any induced binary tree of height $h$ will be stored in \bigoh{1+\frac{2^h}{B^{(1-\epsilon)/\log 3} }} blocks. 
\end{lem}
\begin{proof}
There is a height $h'$ in the range $\log_3 B$ to $\log_3 B^{1-\epsilon}$ whereby all induced subtrees of nodes at that height will fit into a block. A binary tree of height $h$ will have $2^{h-h'}$ trees of height $h'$, and $2^{h-h'}-1$ nodes of height greater than $h'$. Each tree of height $h'$ is stored in memory in $B$ consecutive locations and therefore intersects at most two blocks, thus, even if each node above height $h'$ is in a different block, the total number of blocks the subtree is stored in is $ \bigoh{2^{h-h'}}
=\bigoh{2^{h-\log_3  B^{1-\epsilon}}}
=\varbigoh{\frac{2^h}{2^{(1-\epsilon)\log_3 B} }}
=\varbigoh{\frac{2^h}{B^{(1-\epsilon)\log_3 2} }}
=\varbigoh{\frac{2^h}{B^{(1-\epsilon)/\log 3} }}
$.
\end{proof}

The bottom trees are stored in a more compressed manner; the nodes appear in the same order as if they were in the top tree, but instead of storing all nodes in a complete 3-ary tree, only those nodes that are actually in the bottom tree are stored. Thus the size of the bottom tree is simply the number of nodes in the tree.  The facts presented for the top trees and Lemma~\ref{lem:binarysubgraph} also hold for the compressed representation.

\subsection{Executing operations}

\subsubsection{\readop{i}} We simply return $C[i]$.

\subsubsection{\preadop{v}{i}} The answer is in the tree leaf containing $(i,v)$, the point associated with this operation. We find this leaf by moving $p$ in the obvious way: move the pointer $p$ to parent nodes until you reach a rectangle containing $(i,v)$ or the root of a tree. If you reach a root, set $p$ the root of the appropriate tree, the $\lfloor v/U \rfloor$th one. Then move the pointer $p$ down until you hit a leaf. The answer is in the leaf.

\subsubsection{\writeop{i}{x}}
We call a node a top node if its rectangle is open. 
We will maintain the invariant that no top node is full. Since \kwwrite{} modifies $P$ by adding a point above all others, this guarantees that the new point will intersect only non-full nodes.

The insertion begins by checking for two special cases. The first is if the memory location $i$ is larger than $U$, which is an upper bound on the highest memory location written so far. In this case, we apply the standard doubling trick and $U$ is doubled, all the trees are destroyed and re-built by re-executing all \kwwrite{}s which are stored in the log as mentioned in Section~\ref{subsec:globalds}. The current memory array $C$ is also doubled in size. 

The other special case is when once every $U$ operations the point associated with the \kwwrite{} is on the $(U+1)$-th row of the rectangle of the top tree and a new top tree is needed. 
In this case, the representation of the existing top tree is compressed and copied as a new bottom tree, removing any unused memory locations, and the top tree is reinitialized as a binary tree where the leafs contain values stored in the corresponding cells of $C$.

Then the main part of the \kwwrite{} operation proceeds as follows.

\begin{itemize}
\item Sets $C[i]=x$
\item Increments $V$
\item Follows the pointer to the leaf containing $(i,V)$. Note that this is a top node. Add the data to this leaf and mark it full. 
This means that the point set $P$ now contains the new \kwwrite{}, as it must. However, it is a top node and is full, which violates the previously stated invariant. We then use the following general reconfiguration to preserve that fact that top nodes can not be full.

\end{itemize}
\paragraph{Reconfiguration.}
When a node becomes full we mark it as such and proceed according to one of the following two cases. 
1) The parent of the node already has three children, in which case the parent also becomes full. We recurse with the parent. 
2) The parent of the node has only two children, which implies the parent is still not full. At this point, we close all open rectangles in the subtree of the current node and add a third child to the parent.
(These procedures are explained below in more detail.)
The rectangle of the new child is on top of the rectangle of the current node. 

In other words, when a leaf node becomes full this leads to that leaf node and 0 or more of its ancestors becoming full. We recurse until the highest such ancestor node $p$ and close every open rectangle in its subtree, and add as the third child to the parent of $p$ a sibling node whose open rectangle is on top of $p$'s newly closed rectangle. 
See Figure~\ref{fig:insseq}.
\input{figure2}

\paragraph{Closing rectangles.} 
We close the open rectangles in the subtree of a node by traversing the 2 children of each node that correspond to open rectangles and changing the top side of each open rectangle from $\infty$ to $V$. 

\paragraph{Adding the third child.}
In order to add a third child to a node at height $h+1$, we create a complete binary tree of height $h$ whose root becomes the third child. Note that the leafs of this tree contain the answers to the \kwpread{} queries at the corresponding points. We copy this information from $C$. 

\begin{lem}
\label{lem:rebuild}
The amortized number of times a node at height $h+1$ gains a third child following an insertion into its subtree is \bigoh{\frac{1}{2^h}}.
\end{lem}
\begin{proof}
We use a simple potential function where adding a third child on top of a node $p$ at height $h$ has cost 1, and
all top nodes at height $l \leq h$ in the subtree of $p$ have a potential of $\frac{1}{2^{h-l}}$ for each full child they have. 

Observe that each step during the handling of an insert, a top node with two full children becomes no longer a top node and it is marked as full. Thus, since the potential difference at the level of $p$ matches the cost, the amortized cost at any other level is zero except for at the leaf level where the amortized cost is $\frac{1}{2^h}$. 
\end{proof}

\input{ozanalysis2}

%% file: intro.tex

\section{Introduction}

\label{sec:intro}

Our result is a general transformation to make a data structure partially persistent in the cache-oblivious model.
We first review both the persistence paradigm and the cache-oblivious model before presenting our result.

\paragraph{Persistence.}
\emph{Persistence} is a fundamental data structuring paradigm whereby operations are allowed not only on the current state of the data structure but also on past states. Being able to efficiently work with the past has become a basic feature of many real world systems, including consumer-oriented software such as Apple's Time Machine\ozc{cite}, as well as being a fundamental tool for algorithm development\ozc{cite}. There are several types of persistence which vary in power as to how new versions can be created from existing ones. 

One can also view persistence as a transformation that extends the operations of an underlying data structure ADT. Data structures, generally speaking, have two types of operations, queries and updates. A query is an operation that does not change the structure, while updates do\footnote{Note that, for example, in self-adjusting structures (e.g.~\cite{DBLP:journals/jacm/SleatorT85,DBLP:journals/jacm/Tarjan75,DBLP:journals/algorithmica/FredmanSST86}) operations considered to be queries are not queries by this definition since rotations are performed to move the searched item to the root, thus changing the structure.}. A \emph{version} is a snapshot of the data structure at the completion of an update, and updates generate new versions. In all variants of persistence queries are allowed in any version, past or present. The simplest form of persistence is \emph{partial persistence} whereby updates can only be performed on the most recently generated version, called the \emph{present}. The ensuing collection of versions models a linear view of time whereby one can execute operations only in the present and have a read-only view back into the past. A more complicated form of persistence is known as \emph{full} persistence and allows new versions to be generated by performing an update operation on any of the previous versions; the new version is branched off from the old one and the relationship among the versions form a tree. The most complicated form of persistence is when the underlying data structure allows the combination of two data structures into one. Such operations are called \emph{meld}-type operations, and if meld-type operations are allowed to be performed on any two versions of the data structure, the result is said to be \emph{confluently persistent}. In confluent persistence, the versions form a DAG. In \cite{DBLP:journals/talg/DemaineIL07}, the retroactive model was introduced, which differed from persistence by allowing insertion and deletion of operations in the past, with the results propagated to the present; this is substantially harder than forking off new versions.

Much of the work so far on persistence has been done in the pointer model. 
The landmark result was when in \cite{DBLP:journals/cacm/SarnakT86,DBLP:journals/jcss/DriscollSST89} it was shown that any pointer based structure with constant in-degree can be made fully persistent with no additional asymptotic cost in time and space; the space used was simply the total number of changes to the structure. The first attempt at supporting confluent operations was \cite{DBLP:journals/jal/FiatK03}, but the query runtimes could degrade significantly as a result of their transformation. This was largely because by merging a structure with itself, one can obtain an implicit representation of an exponential sized structure in linear steps. In \cite{DBLP:conf/soda/ColletteIL12} the situation was improved by showing that if a piece of data can only occur once in each structure, then the process of making a structure confluently persistent becomes reasonable.

\paragraph{The cache-oblivious model.}
The \emph{Disk-Access Model
  (DAM)} 
is the classic model of a two-level memory hierarchy; there is a
computer with internal memory size $M$, and an external memory
(typically called for historical reasons the \emph{disk}), and data
can be moved back and forth between internal memory and disk in blocks
of a given size $B$ at unit cost. The underlying premise is that since
disk is so much slower than internal memory, counting the number of
disk-block transfers, while ignoring all else, is a good model of
runtime. A classic data structure in the DAM model is the B-tree
\cite{DBLP:journals/acta/BayerM72}.

However as the modern computer has evolved, it has become not just a
two-level hierarchy but a multi-level hierarchy, ranging from the
registers to the main storage, with several levels of cache in
between. Each level has a smaller amount of faster memory than the
previous one. The most successful attempt to cope with this hierarchy
to date has been employing \emph{cache-oblivious
  algorithms}~\cite{DBLP:journals/talg/FrigoLPR12}, which are not
parameterized by $M$ and $B$.  These algorithms are still analyzed in
a two-level memory hierarchy like the DAM, but the values of $M$ and
$B$ are used only in the analysis and not known to the algorithm.  By
doing the analysis for an arbitrary two-level memory hierarchy, it
applies automatically to all the levels of a multi-level memory
hierarchy.

Cache-oblivious algorithms are analyzed in the ideal-cache
model~\cite{DBLP:journals/talg/FrigoLPR12}, which is the same as the
DAM model but for one addition: the system automatically and optimally
decides which block to evict from the internal memory when loading a
new block.  We use the term \emph{cache-oblivious model} to
encapsulate both the ideal-cache model and the restriction that the
algorithm be cache oblivious.  Some DAM-model algorithms are trivially
cache oblivious, for example scanning an array of size $N$ has cost
$\bigoh{N/B+1}$ in both models. Other structures, such as the $B$-tree, are
parametrized on the value of $B$ and thus cache-oblivious alternatives
to the $B$-tree require completely different methods
(e.g.~\cite{DBLP:journals/siamcomp/BenderDF05,DBLP:journals/jal/BenderDIW04}).

\paragraph{Previous persistence results for memory models.} 

In~\cite{DBLP:conf/soda/BrodalTST12} a general method to make a
pointer-based DAM model algorithm fully persistent (assuming constant
in-degree) was presented where updates incur a $\bigoh{\log B}$
slowdown compared to the ephemeral version. 
In \cite{DBLP:conf/icalp/BenderCR02} a partially persistent dictionary for the cache-oblivious model was presented, with a runtime of $\bigoh{\log_B (\vers+N)}$ for all queries and updates, where $N$ is the size of the dictionary and \vers{} is the number of updates.
Both of these results are based on adapting the methods of the original persistence result for the memory model at hand, and are still tied to the view of a structure as being pointer-based. In contrast, we make no assumptions about a pointer-based structure and do not use any of the techniques from~\cite{DBLP:journals/cacm/SarnakT86}.

\paragraph{Persistence in the cache-oblivious model.}
We present a general transformation to make any cache-oblivious data structure partially persistent; this is the first such general result.

In a cache-oblivious algorithm, from the point of view of the algorithm, the only thing that is manipulated is the lowest level of memory (the disk). The algorithm simply acts as if there is one memory array which is initially empty, and the only thing the algorithm can do is look at and change memory cells. 
As the \emph{ideal cache assumption} states that the optimal offline paging can be approximated within a constant factor online, the algorithm does not directly do any memory management; it just looks at and changes the memory array. 
Working directly with the memory array as a whole is a necessary feature of efficient cache-oblivious algorithms, and the overriding design paradigm is to ensure data locality at all levels of granularity.
We define the space used by a cache-oblivious algorithm, \asize, to be the highest address in memory touched so far. 
The runtime of any operation on a particular state of a data structure is expressed as a function of $B$ and $M$ (only), which are unknown to the algorithm and in fact are different for different levels of the memory hierarchy. 
In partial persistence, another parameter of interest is the number
$\vers$ of updates performed. This is simply defined as the total number of memory cells changed by the algorithm. So, the goal is to 
keep a history of all previous states of the data structure, and support queries in the past
for any cache-oblivious structure while minimizing any additional space and time needed.

As a warm up, look at two simple methods to implement cache-oblivious persistence. If one were to copy the entire memory with every update operation one would be able to achieve partial persistence, but with unacceptable space usage and update cost (since this method would have to spend $\bigoh{\asize/B}$ time copying the memory for every update). On the other extreme, one could store for each location in memory a dynamic predecessor-query structure that records all the values held by that cell in memory. While this approach brings down the space usage, it destroys any locality in memory that the algorithm tried to create; for example, an array scan will cost $\bigoh{\asize}$ rather than $\bigoh{\asize/B}$, as array elements which were sequential in the ephemeral structure would no longer be stored sequentially.

\paragraph{Our results.}
The most straightforward way to present our results is in the case when the ephemeral  data structure's runtime does not assume more than a constant sized memory, and does not benefit from keeping anything in memory between operations. These assumptions often correspond to query-based structures, where blocks that happen to be in memory from previous queries are not of use, and are often reflected in the fact that there is no $M$ in the runtime of the structure. In this case our results are as follows: queries in the present take no additional asymptotic cost; queries in the past and updates incur a $\bigoh{\log_{B} \asize}$-factor time overhead over the cost if they were executed with a polynomially smaller block size.
In particular, a runtime of $\querytimeb{B}$ in the ephemeral structure will become $\bigoh{\querytimebthet{B^{\miniexpin}}\log_{B} \gridwidth}$ for any constant $\epsilon > 0$. Thus, as a simple example, a structure of size $\Theta(N)$ and runtime $\bigoh{\log_BN}$ will support queries in the present in time $\bigoh{\log_BN}$, persistent queries in the past and updates in time $\bigoh{\log_{B}^2 N}$.

Structures where memory plays a role are more involved to analyze. There are two cases. One is where memory is needed only within an operation, but the data structure's runtime is not dependent on anything being left in memory from one operation to the next operation. In this case, as in the above, queries in the present will be run as asymptotically fast as in the ephemeral using a memory of size a constant factor smaller. For queries in the past and updates, our transformation will have the effect of using the runtime of the ephemeral where the memory size is reduced by a $\bigoh{B^{1-\miniexpin}\log_{B} \asize}$ and a \bigoh{B^{1-\miniexpin}} factor, respectively. Thus, in this case, a query of time $\querytimemb{M}{B}$ in the ephemeral structure will take time $\bigoh{\querytimembthet{M}{B}}$ for a query in the present, 
$\bigoh{\querytimembthet{M/B^{1-\miniexpin}\log_{B} \asize}{B^{\miniexpin}}\log_{B}\asize}$ for a query in the past, and $\bigoh{\querytimembthet{M/B^{1-\miniexpin}}{B^{\miniexpin}}\log_{B}\asize}$ for an update.

When the runtime of a structure depends on keeping something in memory between operations, the runtimes of the previous paragraph hold for updates and queries in the present. But, for queries in the past, this is not possible. It is impossible to store the memory of all previous versions in memory. For queries in the past, suppose that executing a sequence of queries in the ephemeral takes time $\querytimemb{M}{B}$ if memory was initially empty. By treating the sequence  as one meta query, 
the results of the previous paragraph apply, and 
our structure will execute the sequence in $\bigoh{\querytimembthet{M/B^{1-\miniexpin}\log_{B} \asize}{B^{\miniexpin}}\log_{B} \asize}$ time. Having $\querytimemb{M}{B}$ represent the runtime in the ephemeral assuming memory starts off empty could cause $\querytimemb{M}{B}$ to be higher than without this requirement. Short sequences of operations that have zero or sub-constant cost will have the largest impact, while long or expensive sequences will not.

Data structures often have runtimes that are a function of a single
variable, commonly denoted by $N$, representing the size of the
structure.  This letter makes no appearance in our construction or
analysis, however, as it would be too restrictive---bounds that
include multiple variables (e.g.~graph algorithms) or more complicated
competitive or instance-based analysis can all be accommodated in the
framework of our result. For the same reason we treat space separately
and can freely work with structures with non-linear space or whose
space usage is not neatly expressed as a function of one variable.

All of our update times are amortized, as doubling tricks are used in
the construction. The amortization cannot be removed easily with
standard lazy rebuilding methods, as in the cache-oblivious model one
cannot run an algorithm for a certain amount of time as this would
imply knowledge of $B$.

A lower bound on the space usage is the total number $\vers$ of memory changes that the structure must remember. If the number of updates is sufficiently long compared to ephemeral space usage (in particular, if $V=\Omega(\asize^{\log 3})$), then our construction has a space usage of $\bigoh{V\log \asize}$. 

\paragraph{Overview of methods.}

The basic design principle of a cache-oblivious data structure is to try to ensure as much locality as possible in executing operations. So, our structure tries to  maintain locality. It does so directly by ensuring that the data of any consecutive range of memory of size $w$ of any past version will be stored in a constant number of consecutive ranges of memory of size $\bigoh{w^{\log 3}}$. Our structure is based on viewing changes to memory as points in two dimensions, with one dimension being the memory address, and the other the time the change in memory was made. With this view, our structure is a decomposition of 2-D space-time, with a particular embedding into memory, and routines to update and query.

%% file: figure1.tex
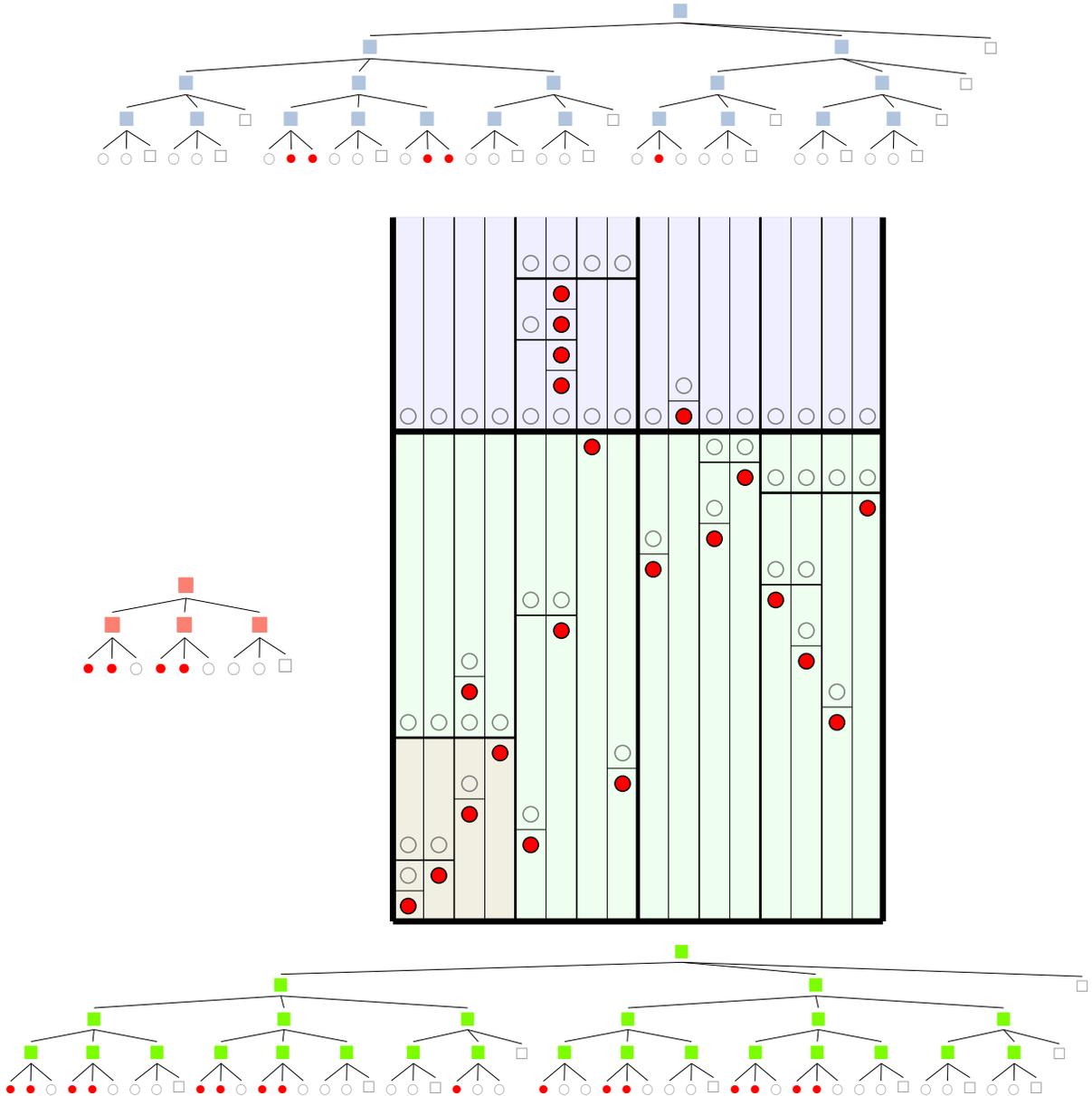
\begin{figure}[!ht]
 \center
\begin{tikzpicture}[scale=0.50]
\Tree
[.{{\LARGE\textcolor{LightSteelBlue}{$\blacksquare$}}} 
[.{{\LARGE\textcolor{LightSteelBlue}{$\blacksquare$}}} 
[.{{\LARGE\textcolor{LightSteelBlue}{$\blacksquare$}}} 
[.{{\LARGE\textcolor{LightSteelBlue}{$\blacksquare$}}} 
[.{{\small\textcolor{Gray}{$\bigcirc$}}} ]
[.{{\small\textcolor{Gray}{$\bigcirc$}}} ]
[.{{\Large\textcolor{Gray}{$\square$}}} ]]
[.{{\LARGE\textcolor{LightSteelBlue}{$\blacksquare$}}} 
[.{{\small\textcolor{Gray}{$\bigcirc$}}} ]
[.{{\small\textcolor{Gray}{$\bigcirc$}}} ]
[.{{\Large\textcolor{Gray}{$\square$}}} ]]
[.{{\Large\textcolor{Gray}{$\square$}}} ]]
[.{{\LARGE\textcolor{LightSteelBlue}{$\blacksquare$}}} 
[.{{\LARGE\textcolor{LightSteelBlue}{$\blacksquare$}}} 
[.{{\small\textcolor{Gray}{$\bigcirc$}}} ]
[.{\textcolor{Red}{$\CIRCLE$}} ]
[.{\textcolor{Red}{$\CIRCLE$}} ]]
[.{{\LARGE\textcolor{LightSteelBlue}{$\blacksquare$}}} 
[.{{\small\textcolor{Gray}{$\bigcirc$}}} ]
[.{{\small\textcolor{Gray}{$\bigcirc$}}} ]
[.{{\Large\textcolor{Gray}{$\square$}}} ]]
[.{{\LARGE\textcolor{LightSteelBlue}{$\blacksquare$}}} 
[.{{\small\textcolor{Gray}{$\bigcirc$}}} ]
[.{\textcolor{Red}{$\CIRCLE$}} ]
[.{\textcolor{Red}{$\CIRCLE$}} ]]]
[.{{\LARGE\textcolor{LightSteelBlue}{$\blacksquare$}}} 
[.{{\LARGE\textcolor{LightSteelBlue}{$\blacksquare$}}} 
[.{{\small\textcolor{Gray}{$\bigcirc$}}} ]
[.{{\small\textcolor{Gray}{$\bigcirc$}}} ]
[.{{\Large\textcolor{Gray}{$\square$}}} ]]
[.{{\LARGE\textcolor{LightSteelBlue}{$\blacksquare$}}} 
[.{{\small\textcolor{Gray}{$\bigcirc$}}} ]
[.{{\small\textcolor{Gray}{$\bigcirc$}}} ]
[.{{\Large\textcolor{Gray}{$\square$}}} ]]
[.{{\Large\textcolor{Gray}{$\square$}}} ]]]
[.{{\LARGE\textcolor{LightSteelBlue}{$\blacksquare$}}} 
[.{{\LARGE\textcolor{LightSteelBlue}{$\blacksquare$}}} 
[.{{\LARGE\textcolor{LightSteelBlue}{$\blacksquare$}}} 
[.{{\small\textcolor{Gray}{$\bigcirc$}}} ]
[.{\textcolor{Red}{$\CIRCLE$}} ]
[.{{\small\textcolor{Gray}{$\bigcirc$}}} ]]
[.{{\LARGE\textcolor{LightSteelBlue}{$\blacksquare$}}} 
[.{{\small\textcolor{Gray}{$\bigcirc$}}} ]
[.{{\small\textcolor{Gray}{$\bigcirc$}}} ]
[.{{\Large\textcolor{Gray}{$\square$}}} ]]
[.{{\Large\textcolor{Gray}{$\square$}}} ]]
[.{{\LARGE\textcolor{LightSteelBlue}{$\blacksquare$}}} 
[.{{\LARGE\textcolor{LightSteelBlue}{$\blacksquare$}}} 
[.{{\small\textcolor{Gray}{$\bigcirc$}}} ]
[.{{\small\textcolor{Gray}{$\bigcirc$}}} ]
[.{{\Large\textcolor{Gray}{$\square$}}} ]]
[.{{\LARGE\textcolor{LightSteelBlue}{$\blacksquare$}}} 
[.{{\small\textcolor{Gray}{$\bigcirc$}}} ]
[.{{\small\textcolor{Gray}{$\bigcirc$}}} ]
[.{{\Large\textcolor{Gray}{$\square$}}} ]]
[.{{\Large\textcolor{Gray}{$\square$}}} ]]
[.{{\Large\textcolor{Gray}{$\square$}}} ]]
[.{{\Large\textcolor{Gray}{$\square$}}} ]]
\end{tikzpicture}\vspace{2mm}
\linebreak
\begin{minipage}[t]{0.25\linewidth}\vspace{150pt}
\begin{tikzpicture}[scale=0.55]
\Tree
[.{{\LARGE\textcolor{Salmon}{$\blacksquare$}}} 
[.{{\LARGE\textcolor{Salmon}{$\blacksquare$}}} 
[.{\textcolor{Red}{$\CIRCLE$}} ]
[.{\textcolor{Red}{$\CIRCLE$}} ]
[.{{\small\textcolor{Gray}{$\bigcirc$}}} ]]
[.{{\LARGE\textcolor{Salmon}{$\blacksquare$}}} 
[.{\textcolor{Red}{$\CIRCLE$}} ]
[.{\textcolor{Red}{$\CIRCLE$}} ]
[.{{\small\textcolor{Gray}{$\bigcirc$}}} ]]
[.{{\LARGE\textcolor{Salmon}{$\blacksquare$}}} 
[.{{\small\textcolor{Gray}{$\bigcirc$}}} ]
[.{{\small\textcolor{Gray}{$\bigcirc$}}} ]
[.{{\Large\textcolor{Gray}{$\square$}}} ]]]
\end{tikzpicture}
\end{minipage}  \hspace*{.2cm}  
\begin{minipage}[t]{0.55\linewidth}\vspace{0pt}
\begin{tikzpicture}[scale=0.45]
\draw [blue!40, fill, opacity=.15] (0,16) rectangle (16,23);
\draw [green!40, fill, opacity=.15] (0,0) rectangle (16,16);
\draw [red!40, fill, opacity=.15] (0,0) rectangle (4,6);
\draw [line width=2.5732799] (0,0)--(16,0);
\draw [line width=2.5732799] (0,0)--(0,23);
\draw [line width=2.5732799] (16,0)--(16,23);
\draw [line width=0.3,black] (1,0)--(1,23);
\draw [line width=0.3,black] (2,0)--(2,23);
\draw [line width=0.3,black] (3,0)--(3,23);
\draw [line width=0.3,black] (4,0)--(4,23);
\draw [line width=0.3,black] (5,0)--(5,23);
\draw [line width=0.3,black] (6,0)--(6,23);
\draw [line width=0.3,black] (7,0)--(7,23);
\draw [line width=0.3,black] (8,0)--(8,23);
\draw [line width=0.3,black] (9,0)--(9,23);
\draw [line width=0.3,black] (10,0)--(10,23);
\draw [line width=0.3,black] (11,0)--(11,23);
\draw [line width=0.3,black] (12,0)--(12,23);
\draw [line width=0.3,black] (13,0)--(13,23);
\draw [line width=0.3,black] (14,0)--(14,23);
\draw [line width=0.3,black] (15,0)--(15,23);
\draw [line width=0.62,black] (2,0)--(2,23);
\draw [line width=0.62,black] (4,0)--(4,23);
\draw [line width=0.62,black] (6,0)--(6,23);
\draw [line width=0.62,black] (8,0)--(8,23);
\draw [line width=0.62,black] (10,0)--(10,23);
\draw [line width=0.62,black] (12,0)--(12,23);
\draw [line width=0.62,black] (14,0)--(14,23);
\draw [line width=1.068,black] (4,0)--(4,23);
\draw [line width=1.068,black] (8,0)--(8,23);
\draw [line width=1.068,black] (12,0)--(12,23);
\draw [line width=1.6952,black] (8,0)--(8,23);
\draw [line width=0.3,black] (0,1)--(1,1);
\draw [semithick, black, fill=red] (0.5,0.5) circle [radius=0.25];
\draw [semithick, gray] (0.5,1.5) circle [radius=0.25];
\draw [line width=0.62,black] (1,2)--(2,2);
\draw [semithick, black, fill=red] (1.5,1.5) circle [radius=0.25];
\draw [semithick, gray] (1.5,2.5) circle [radius=0.25];
\draw [line width=0.62,black] (0,2)--(1,2);
\draw [semithick, gray] (0.5,2.5) circle [radius=0.25];
\draw [line width=0.3,black] (2,4)--(3,4);
\draw [semithick, black, fill=red] (2.5,3.5) circle [radius=0.25];
\draw [semithick, gray] (2.5,4.5) circle [radius=0.25];
\draw [line width=1.068,black] (3,6)--(4,6);
\draw [semithick, black, fill=red] (3.5,5.5) circle [radius=0.25];
\draw [semithick, gray] (3.5,6.5) circle [radius=0.25];
\draw [line width=1.068,black] (2,6)--(3,6);
\draw [semithick, gray] (2.5,6.5) circle [radius=0.25];
\draw [line width=1.068,black] (0,6)--(1,6);
\draw [semithick, gray] (0.5,6.5) circle [radius=0.25];
\draw [line width=1.068,black] (1,6)--(2,6);
\draw [semithick, gray] (1.5,6.5) circle [radius=0.25];
\draw [line width=0.3,black] (4,3)--(5,3);
\draw [semithick, black, fill=red] (4.5,2.5) circle [radius=0.25];
\draw [semithick, gray] (4.5,3.5) circle [radius=0.25];
\draw [line width=0.62,black] (5,10)--(6,10);
\draw [semithick, black, fill=red] (5.5,9.5) circle [radius=0.25];
\draw [semithick, gray] (5.5,10.5) circle [radius=0.25];
\draw [line width=0.62,black] (4,10)--(5,10);
\draw [semithick, gray] (4.5,10.5) circle [radius=0.25];
\draw [line width=2.5732799,black] (6,16)--(7,16);
\draw [semithick, black, fill=red] (6.5,15.5) circle [radius=0.25];
\draw [semithick, gray] (6.5,16.5) circle [radius=0.25];
\draw [line width=0.3,black] (7,5)--(8,5);
\draw [semithick, black, fill=red] (7.5,4.5) circle [radius=0.25];
\draw [semithick, gray] (7.5,5.5) circle [radius=0.25];
\draw [line width=2.5732799,black] (7,16)--(8,16);
\draw [semithick, gray] (7.5,16.5) circle [radius=0.25];
\draw [line width=2.5732799,black] (4,16)--(5,16);
\draw [semithick, gray] (4.5,16.5) circle [radius=0.25];
\draw [line width=2.5732799,black] (5,16)--(6,16);
\draw [semithick, gray] (5.5,16.5) circle [radius=0.25];
\draw [line width=2.5732799,black] (0,16)--(1,16);
\draw [semithick, gray] (0.5,16.5) circle [radius=0.25];
\draw [line width=2.5732799,black] (1,16)--(2,16);
\draw [semithick, gray] (1.5,16.5) circle [radius=0.25];
\draw [line width=0.3,black] (2,8)--(3,8);
\draw [semithick, black, fill=red] (2.5,7.5) circle [radius=0.25];
\draw [semithick, gray] (2.5,8.5) circle [radius=0.25];
\draw [line width=2.5732799,black] (3,16)--(4,16);
\draw [semithick, gray] (3.5,16.5) circle [radius=0.25];
\draw [line width=2.5732799,black] (2,16)--(3,16);
\draw [semithick, gray] (2.5,16.5) circle [radius=0.25];
\draw [line width=0.3,black] (8,12)--(9,12);
\draw [semithick, black, fill=red] (8.5,11.5) circle [radius=0.25];
\draw [semithick, gray] (8.5,12.5) circle [radius=0.25];
\draw [line width=2.5732799,black] (9,16)--(10,16);
\draw [semithick, gray] (9.5,16.5) circle [radius=0.25];
\draw [line width=2.5732799,black] (8,16)--(9,16);
\draw [semithick, gray] (8.5,16.5) circle [radius=0.25];
\draw [line width=0.3,black] (10,13)--(11,13);
\draw [semithick, black, fill=red] (10.5,12.5) circle [radius=0.25];
\draw [semithick, gray] (10.5,13.5) circle [radius=0.25];
\draw [line width=0.62,black] (11,15)--(12,15);
\draw [semithick, black, fill=red] (11.5,14.5) circle [radius=0.25];
\draw [semithick, gray] (11.5,15.5) circle [radius=0.25];
\draw [line width=0.62,black] (10,15)--(11,15);
\draw [semithick, gray] (10.5,15.5) circle [radius=0.25];
\draw [line width=2.5732799,black] (10,16)--(11,16);
\draw [semithick, gray] (10.5,16.5) circle [radius=0.25];
\draw [line width=2.5732799,black] (11,16)--(12,16);
\draw [semithick, gray] (11.5,16.5) circle [radius=0.25];
\draw [line width=0.62,black] (12,11)--(13,11);
\draw [semithick, black, fill=red] (12.5,10.5) circle [radius=0.25];
\draw [semithick, gray] (12.5,11.5) circle [radius=0.25];
\draw [line width=0.3,black] (13,9)--(14,9);
\draw [semithick, black, fill=red] (13.5,8.5) circle [radius=0.25];
\draw [semithick, gray] (13.5,9.5) circle [radius=0.25];
\draw [line width=0.62,black] (13,11)--(14,11);
\draw [semithick, gray] (13.5,11.5) circle [radius=0.25];
\draw [line width=0.3,black] (14,7)--(15,7);
\draw [semithick, black, fill=red] (14.5,6.5) circle [radius=0.25];
\draw [semithick, gray] (14.5,7.5) circle [radius=0.25];
\draw [line width=1.068,black] (15,14)--(16,14);
\draw [semithick, black, fill=red] (15.5,13.5) circle [radius=0.25];
\draw [semithick, gray] (15.5,14.5) circle [radius=0.25];
\draw [line width=1.068,black] (14,14)--(15,14);
\draw [semithick, gray] (14.5,14.5) circle [radius=0.25];
\draw [line width=1.068,black] (12,14)--(13,14);
\draw [semithick, gray] (12.5,14.5) circle [radius=0.25];
\draw [line width=1.068,black] (13,14)--(14,14);
\draw [semithick, gray] (13.5,14.5) circle [radius=0.25];
\draw [line width=2.5732799,black] (12,16)--(13,16);
\draw [semithick, gray] (12.5,16.5) circle [radius=0.25];
\draw [line width=2.5732799,black] (13,16)--(14,16);
\draw [semithick, gray] (13.5,16.5) circle [radius=0.25];
\draw [line width=2.5732799,black] (14,16)--(15,16);
\draw [semithick, gray] (14.5,16.5) circle [radius=0.25];
\draw [line width=2.5732799,black] (15,16)--(16,16);
\draw [semithick, gray] (15.5,16.5) circle [radius=0.25];
\draw [line width=0.62,black] (4,19)--(5,19);
\draw [semithick, gray] (4.5,19.5) circle [radius=0.25];
\draw [line width=0.3,black] (5,18)--(6,18);
\draw [semithick, black, fill=red] (5.5,17.5) circle [radius=0.25];
\draw [semithick, gray] (5.5,18.5) circle [radius=0.25];
\draw [line width=0.62,black] (5,19)--(6,19);
\draw [semithick, black, fill=red] (5.5,18.5) circle [radius=0.25];
\draw [semithick, gray] (5.5,19.5) circle [radius=0.25];
\draw [line width=1.068,black] (6,21)--(7,21);
\draw [semithick, gray] (6.5,21.5) circle [radius=0.25];
\draw [line width=1.068,black] (7,21)--(8,21);
\draw [semithick, gray] (7.5,21.5) circle [radius=0.25];
\draw [line width=1.068,black] (4,21)--(5,21);
\draw [semithick, gray] (4.5,21.5) circle [radius=0.25];
\draw [line width=0.3,black] (5,20)--(6,20);
\draw [semithick, black, fill=red] (5.5,19.5) circle [radius=0.25];
\draw [semithick, gray] (5.5,20.5) circle [radius=0.25];
\draw [line width=1.068,black] (5,21)--(6,21);
\draw [semithick, black, fill=red] (5.5,20.5) circle [radius=0.25];
\draw [semithick, gray] (5.5,21.5) circle [radius=0.25];
\draw [line width=0.3,black] (9,17)--(10,17);
\draw [semithick, black, fill=red] (9.5,16.5) circle [radius=0.25];
\draw [semithick, gray] (9.5,17.5) circle [radius=0.25];
\end{tikzpicture}
\end{minipage}   
\vspace{2mm}\linebreak
\begin{tikzpicture}[scale=0.47]
\Tree
[.{{\LARGE\textcolor{LawnGreen}{$\blacksquare$}}} 
[.{{\LARGE\textcolor{LawnGreen}{$\blacksquare$}}} 
[.{{\LARGE\textcolor{LawnGreen}{$\blacksquare$}}} 
[.{{\LARGE\textcolor{LawnGreen}{$\blacksquare$}}} 
[.{\textcolor{Red}{$\CIRCLE$}} ]
[.{\textcolor{Red}{$\CIRCLE$}} ]
[.{{\small\textcolor{Gray}{$\bigcirc$}}} ]]
[.{{\LARGE\textcolor{LawnGreen}{$\blacksquare$}}} 
[.{\textcolor{Red}{$\CIRCLE$}} ]
[.{\textcolor{Red}{$\CIRCLE$}} ]
[.{{\small\textcolor{Gray}{$\bigcirc$}}} ]]
[.{{\LARGE\textcolor{LawnGreen}{$\blacksquare$}}} 
[.{{\small\textcolor{Gray}{$\bigcirc$}}} ]
[.{{\small\textcolor{Gray}{$\bigcirc$}}} ]
[.{{\Large\textcolor{Gray}{$\square$}}} ]]]
[.{{\LARGE\textcolor{LawnGreen}{$\blacksquare$}}} 
[.{{\LARGE\textcolor{LawnGreen}{$\blacksquare$}}} 
[.{\textcolor{Red}{$\CIRCLE$}} ]
[.{\textcolor{Red}{$\CIRCLE$}} ]
[.{{\small\textcolor{Gray}{$\bigcirc$}}} ]]
[.{{\LARGE\textcolor{LawnGreen}{$\blacksquare$}}} 
[.{\textcolor{Red}{$\CIRCLE$}} ]
[.{\textcolor{Red}{$\CIRCLE$}} ]
[.{{\small\textcolor{Gray}{$\bigcirc$}}} ]]
[.{{\LARGE\textcolor{LawnGreen}{$\blacksquare$}}} 
[.{{\small\textcolor{Gray}{$\bigcirc$}}} ]
[.{{\small\textcolor{Gray}{$\bigcirc$}}} ]
[.{{\Large\textcolor{Gray}{$\square$}}} ]]]
[.{{\LARGE\textcolor{LawnGreen}{$\blacksquare$}}} 
[.{{\LARGE\textcolor{LawnGreen}{$\blacksquare$}}} 
[.{{\small\textcolor{Gray}{$\bigcirc$}}} ]
[.{{\small\textcolor{Gray}{$\bigcirc$}}} ]
[.{{\Large\textcolor{Gray}{$\square$}}} ]]
[.{{\LARGE\textcolor{LawnGreen}{$\blacksquare$}}} 
[.{\textcolor{Red}{$\CIRCLE$}} ]
[.{{\small\textcolor{Gray}{$\bigcirc$}}} ]
[.{{\small\textcolor{Gray}{$\bigcirc$}}} ]]
[.{{\Large\textcolor{Gray}{$\square$}}} ]]]
[.{{\LARGE\textcolor{LawnGreen}{$\blacksquare$}}} 
[.{{\LARGE\textcolor{LawnGreen}{$\blacksquare$}}} 
[.{{\LARGE\textcolor{LawnGreen}{$\blacksquare$}}} 
[.{\textcolor{Red}{$\CIRCLE$}} ]
[.{{\small\textcolor{Gray}{$\bigcirc$}}} ]
[.{{\small\textcolor{Gray}{$\bigcirc$}}} ]]
[.{{\LARGE\textcolor{LawnGreen}{$\blacksquare$}}} 
[.{\textcolor{Red}{$\CIRCLE$}} ]
[.{\textcolor{Red}{$\CIRCLE$}} ]
[.{{\small\textcolor{Gray}{$\bigcirc$}}} ]]
[.{{\LARGE\textcolor{LawnGreen}{$\blacksquare$}}} 
[.{{\small\textcolor{Gray}{$\bigcirc$}}} ]
[.{{\small\textcolor{Gray}{$\bigcirc$}}} ]
[.{{\Large\textcolor{Gray}{$\square$}}} ]]]
[.{{\LARGE\textcolor{LawnGreen}{$\blacksquare$}}} 
[.{{\LARGE\textcolor{LawnGreen}{$\blacksquare$}}} 
[.{\textcolor{Red}{$\CIRCLE$}} ]
[.{\textcolor{Red}{$\CIRCLE$}} ]
[.{{\small\textcolor{Gray}{$\bigcirc$}}} ]]
[.{{\LARGE\textcolor{LawnGreen}{$\blacksquare$}}} 
[.{\textcolor{Red}{$\CIRCLE$}} ]
[.{\textcolor{Red}{$\CIRCLE$}} ]
[.{{\small\textcolor{Gray}{$\bigcirc$}}} ]]
[.{{\LARGE\textcolor{LawnGreen}{$\blacksquare$}}} 
[.{{\small\textcolor{Gray}{$\bigcirc$}}} ]
[.{{\small\textcolor{Gray}{$\bigcirc$}}} ]
[.{{\Large\textcolor{Gray}{$\square$}}} ]]]
[.{{\LARGE\textcolor{LawnGreen}{$\blacksquare$}}} 
[.{{\LARGE\textcolor{LawnGreen}{$\blacksquare$}}} 
[.{{\small\textcolor{Gray}{$\bigcirc$}}} ]
[.{{\small\textcolor{Gray}{$\bigcirc$}}} ]
[.{{\Large\textcolor{Gray}{$\square$}}} ]]
[.{{\LARGE\textcolor{LawnGreen}{$\blacksquare$}}} 
[.{{\small\textcolor{Gray}{$\bigcirc$}}} ]
[.{{\small\textcolor{Gray}{$\bigcirc$}}} ]
[.{{\Large\textcolor{Gray}{$\square$}}} ]]
[.{{\Large\textcolor{Gray}{$\square$}}} ]]]
[.{{\Large\textcolor{Gray}{$\square$}}} ]]
\end{tikzpicture}
\caption{The integer grid representing a sequence of \kwwrite{} operations is shown for $U=16$. Red circles represent points of $P$. Hollow circles represent the answers to the \kwpread{} queries at the bottom of each leaf node's rectangle. Each rectangle is associated with a node in an ST-tree. Three subtrees of the ST-tree are also highlighted. The blue shaded rectangle (and the tree above) corresponds to the top tree. The green (plus red) shaded rectangle (and the tree below) corresponds to the bottom tree. The red shaded rectangle corresponds to the subtree shown to the left of the grid. \\ \\
}

\label{fig:space-time}

\end{figure}

%% file: figure2.tex
\begin{figure}
\center

\vspace{1mm}
\begin{minipage}[t]{0.39\linewidth}\vspace{0pt}
\begin{tikzpicture}[scale=0.45]
\draw [line width=1.6952] (0,0)--(8,0);
\draw [line width=1.6952] (0,0)--(0,7);
\draw [line width=1.6952] (8,0)--(8,7);
\draw [line width=0.3,black] (1,0)--(1,7);
\draw [line width=0.3,black] (2,0)--(2,7);
\draw [line width=0.3,black] (3,0)--(3,7);
\draw [line width=0.3,black] (4,0)--(4,7);
\draw [line width=0.3,black] (5,0)--(5,7);
\draw [line width=0.3,black] (6,0)--(6,7);
\draw [line width=0.3,black] (7,0)--(7,7);
\draw [line width=0.62,black] (2,0)--(2,7);
\draw [line width=0.62,black] (4,0)--(4,7);
\draw [line width=0.62,black] (6,0)--(6,7);
\draw [line width=1.068,black] (4,0)--(4,7);
\draw [line width=0.3,black] (0,1)--(1,1);
\draw [semithick, black, fill=red] (0.5,0.5) circle [radius=0.25];
\draw [semithick, gray] (0.5,1.5) circle [radius=0.25];
\end{tikzpicture}
\end{minipage}   
\vspace{1mm}
\begin{minipage}[t]{0.59\linewidth}\vspace{0pt}
\begin{tikzpicture}[scale=0.6]
\Tree
[.{{\LARGE\textcolor{LightSteelBlue}{$\blacksquare$}}} 
[.{{\LARGE\textcolor{LightSteelBlue}{$\blacksquare$}}} 
[.{{\LARGE\textcolor{LightSteelBlue}{$\blacksquare$}}} 
[.{\textcolor{Red}{$\CIRCLE$}} ]
[.{{\small\textcolor{Gray}{$\bigcirc$}}} ]
[.{{\small\textcolor{Gray}{$\bigcirc$}}} ]]
[.{{\LARGE\textcolor{LightSteelBlue}{$\blacksquare$}}} 
[.{{\small\textcolor{Gray}{$\bigcirc$}}} ]
[.{{\small\textcolor{Gray}{$\bigcirc$}}} ]
[.{{\Large\textcolor{Gray}{$\square$}}} ]]
[.{{\Large\textcolor{Gray}{$\square$}}} ]]
[.{{\LARGE\textcolor{LightSteelBlue}{$\blacksquare$}}} 
[.{{\LARGE\textcolor{LightSteelBlue}{$\blacksquare$}}} 
[.{{\small\textcolor{Gray}{$\bigcirc$}}} ]
[.{{\small\textcolor{Gray}{$\bigcirc$}}} ]
[.{{\Large\textcolor{Gray}{$\square$}}} ]]
[.{{\LARGE\textcolor{LightSteelBlue}{$\blacksquare$}}} 
[.{{\small\textcolor{Gray}{$\bigcirc$}}} ]
[.{{\small\textcolor{Gray}{$\bigcirc$}}} ]
[.{{\Large\textcolor{Gray}{$\square$}}} ]]
[.{{\Large\textcolor{Gray}{$\square$}}} ]]
[.{{\Large\textcolor{Gray}{$\square$}}} ]]
\end{tikzpicture}
\end{minipage}

\vspace{1mm}
\begin{minipage}[t]{0.39\linewidth}\vspace{0pt}
\begin{tikzpicture}[scale=0.45]
\draw [line width=1.6952] (0,0)--(8,0);
\draw [line width=1.6952] (0,0)--(0,7);
\draw [line width=1.6952] (8,0)--(8,7);
\draw [line width=0.3,black] (1,0)--(1,7);
\draw [line width=0.3,black] (2,0)--(2,7);
\draw [line width=0.3,black] (3,0)--(3,7);
\draw [line width=0.3,black] (4,0)--(4,7);
\draw [line width=0.3,black] (5,0)--(5,7);
\draw [line width=0.3,black] (6,0)--(6,7);
\draw [line width=0.3,black] (7,0)--(7,7);
\draw [line width=0.62,black] (2,0)--(2,7);
\draw [line width=0.62,black] (4,0)--(4,7);
\draw [line width=0.62,black] (6,0)--(6,7);
\draw [line width=1.068,black] (4,0)--(4,7);
\draw [line width=0.3,black] (0,1)--(1,1);
\draw [semithick, black, fill=red] (0.5,0.5) circle [radius=0.25];
\draw [semithick, gray] (0.5,1.5) circle [radius=0.25];
\draw [line width=0.3,black] (3,2)--(4,2);
\draw [semithick, black, fill=red] (3.5,1.5) circle [radius=0.25];
\draw [semithick, gray] (3.5,2.5) circle [radius=0.25];
\end{tikzpicture}
\end{minipage}   
\vspace{1mm}
\begin{minipage}[t]{0.59\linewidth}\vspace{0pt}
\begin{tikzpicture}[scale=0.6]
\Tree
[.{{\LARGE\textcolor{LightSteelBlue}{$\blacksquare$}}} 
[.{{\LARGE\textcolor{LightSteelBlue}{$\blacksquare$}}} 
[.{{\LARGE\textcolor{LightSteelBlue}{$\blacksquare$}}} 
[.{\textcolor{Red}{$\CIRCLE$}} ]
[.{{\small\textcolor{Gray}{$\bigcirc$}}} ]
[.{{\small\textcolor{Gray}{$\bigcirc$}}} ]]
[.{{\LARGE\textcolor{LightSteelBlue}{$\blacksquare$}}} 
[.{{\small\textcolor{Gray}{$\bigcirc$}}} ]
[.{\textcolor{Red}{$\CIRCLE$}} ]
[.{{\small\textcolor{Gray}{$\bigcirc$}}} ]]
[.{{\Large\textcolor{Gray}{$\square$}}} ]]
[.{{\LARGE\textcolor{LightSteelBlue}{$\blacksquare$}}} 
[.{{\LARGE\textcolor{LightSteelBlue}{$\blacksquare$}}} 
[.{{\small\textcolor{Gray}{$\bigcirc$}}} ]
[.{{\small\textcolor{Gray}{$\bigcirc$}}} ]
[.{{\Large\textcolor{Gray}{$\square$}}} ]]
[.{{\LARGE\textcolor{LightSteelBlue}{$\blacksquare$}}} 
[.{{\small\textcolor{Gray}{$\bigcirc$}}} ]
[.{{\small\textcolor{Gray}{$\bigcirc$}}} ]
[.{{\Large\textcolor{Gray}{$\square$}}} ]]
[.{{\Large\textcolor{Gray}{$\square$}}} ]]
[.{{\Large\textcolor{Gray}{$\square$}}} ]]
\end{tikzpicture}
\end{minipage}

\vspace{1mm}
\begin{minipage}[t]{0.39\linewidth}\vspace{0pt}
\begin{tikzpicture}[scale=0.45]
\draw [line width=1.6952] (0,0)--(8,0);
\draw [line width=1.6952] (0,0)--(0,7);
\draw [line width=1.6952] (8,0)--(8,7);
\draw [line width=0.3,black] (1,0)--(1,7);
\draw [line width=0.3,black] (2,0)--(2,7);
\draw [line width=0.3,black] (3,0)--(3,7);
\draw [line width=0.3,black] (4,0)--(4,7);
\draw [line width=0.3,black] (5,0)--(5,7);
\draw [line width=0.3,black] (6,0)--(6,7);
\draw [line width=0.3,black] (7,0)--(7,7);
\draw [line width=0.62,black] (2,0)--(2,7);
\draw [line width=0.62,black] (4,0)--(4,7);
\draw [line width=0.62,black] (6,0)--(6,7);
\draw [line width=1.068,black] (4,0)--(4,7);
\draw [line width=0.3,black] (0,1)--(1,1);
\draw [semithick, black, fill=red] (0.5,0.5) circle [radius=0.25];
\draw [semithick, gray] (0.5,1.5) circle [radius=0.25];
\draw [line width=0.62,black] (1,3)--(2,3);
\draw [semithick, black, fill=red] (1.5,2.5) circle [radius=0.25];
\draw [semithick, gray] (1.5,3.5) circle [radius=0.25];
\draw [line width=0.62,black] (0,3)--(1,3);
\draw [semithick, gray] (0.5,3.5) circle [radius=0.25];
\draw [line width=0.3,black] (3,2)--(4,2);
\draw [semithick, black, fill=red] (3.5,1.5) circle [radius=0.25];
\draw [semithick, gray] (3.5,2.5) circle [radius=0.25];
\end{tikzpicture}
\end{minipage}   
\vspace{1mm}
\begin{minipage}[t]{0.59\linewidth}\vspace{0pt}
\begin{tikzpicture}[scale=0.6]
\Tree
[.{{\LARGE\textcolor{LightSteelBlue}{$\blacksquare$}}} 
[.{{\LARGE\textcolor{LightSteelBlue}{$\blacksquare$}}} 
[.{{\LARGE\textcolor{LightSteelBlue}{$\blacksquare$}}} 
[.{\textcolor{Red}{$\CIRCLE$}} ]
[.{\textcolor{Red}{$\CIRCLE$}} ]
[.{{\small\textcolor{Gray}{$\bigcirc$}}} ]]
[.{{\LARGE\textcolor{LightSteelBlue}{$\blacksquare$}}} 
[.{{\small\textcolor{Gray}{$\bigcirc$}}} ]
[.{\textcolor{Red}{$\CIRCLE$}} ]
[.{{\small\textcolor{Gray}{$\bigcirc$}}} ]]
[.{{\LARGE\textcolor{LightSteelBlue}{$\blacksquare$}}} 
[.{{\small\textcolor{Gray}{$\bigcirc$}}} ]
[.{{\small\textcolor{Gray}{$\bigcirc$}}} ]
[.{{\Large\textcolor{Gray}{$\square$}}} ]]]
[.{{\LARGE\textcolor{LightSteelBlue}{$\blacksquare$}}} 
[.{{\LARGE\textcolor{LightSteelBlue}{$\blacksquare$}}} 
[.{{\small\textcolor{Gray}{$\bigcirc$}}} ]
[.{{\small\textcolor{Gray}{$\bigcirc$}}} ]
[.{{\Large\textcolor{Gray}{$\square$}}} ]]
[.{{\LARGE\textcolor{LightSteelBlue}{$\blacksquare$}}} 
[.{{\small\textcolor{Gray}{$\bigcirc$}}} ]
[.{{\small\textcolor{Gray}{$\bigcirc$}}} ]
[.{{\Large\textcolor{Gray}{$\square$}}} ]]
[.{{\Large\textcolor{Gray}{$\square$}}} ]]
[.{{\Large\textcolor{Gray}{$\square$}}} ]]
\end{tikzpicture}
\end{minipage}

\vspace{1mm}
\begin{minipage}[t]{0.39\linewidth}\vspace{0pt}
\begin{tikzpicture}[scale=0.45]
\draw [line width=1.6952] (0,0)--(8,0);
\draw [line width=1.6952] (0,0)--(0,7);
\draw [line width=1.6952] (8,0)--(8,7);
\draw [line width=0.3,black] (1,0)--(1,7);
\draw [line width=0.3,black] (2,0)--(2,7);
\draw [line width=0.3,black] (3,0)--(3,7);
\draw [line width=0.3,black] (4,0)--(4,7);
\draw [line width=0.3,black] (5,0)--(5,7);
\draw [line width=0.3,black] (6,0)--(6,7);
\draw [line width=0.3,black] (7,0)--(7,7);
\draw [line width=0.62,black] (2,0)--(2,7);
\draw [line width=0.62,black] (4,0)--(4,7);
\draw [line width=0.62,black] (6,0)--(6,7);
\draw [line width=1.068,black] (4,0)--(4,7);
\draw [line width=0.3,black] (0,1)--(1,1);
\draw [semithick, black, fill=red] (0.5,0.5) circle [radius=0.25];
\draw [semithick, gray] (0.5,1.5) circle [radius=0.25];
\draw [line width=0.62,black] (1,3)--(2,3);
\draw [semithick, black, fill=red] (1.5,2.5) circle [radius=0.25];
\draw [semithick, gray] (1.5,3.5) circle [radius=0.25];
\draw [line width=0.62,black] (0,3)--(1,3);
\draw [semithick, gray] (0.5,3.5) circle [radius=0.25];
\draw [line width=0.3,black] (3,2)--(4,2);
\draw [semithick, black, fill=red] (3.5,1.5) circle [radius=0.25];
\draw [semithick, gray] (3.5,2.5) circle [radius=0.25];
\draw [line width=0.3,black] (7,4)--(8,4);
\draw [semithick, black, fill=red] (7.5,3.5) circle [radius=0.25];
\draw [semithick, gray] (7.5,4.5) circle [radius=0.25];
\end{tikzpicture}
\end{minipage}   
\vspace{1mm}
\begin{minipage}[t]{0.59\linewidth}\vspace{0pt}
\begin{tikzpicture}[scale=0.6]
\Tree
[.{{\LARGE\textcolor{LightSteelBlue}{$\blacksquare$}}} 
[.{{\LARGE\textcolor{LightSteelBlue}{$\blacksquare$}}} 
[.{{\LARGE\textcolor{LightSteelBlue}{$\blacksquare$}}} 
[.{\textcolor{Red}{$\CIRCLE$}} ]
[.{\textcolor{Red}{$\CIRCLE$}} ]
[.{{\small\textcolor{Gray}{$\bigcirc$}}} ]]
[.{{\LARGE\textcolor{LightSteelBlue}{$\blacksquare$}}} 
[.{{\small\textcolor{Gray}{$\bigcirc$}}} ]
[.{\textcolor{Red}{$\CIRCLE$}} ]
[.{{\small\textcolor{Gray}{$\bigcirc$}}} ]]
[.{{\LARGE\textcolor{LightSteelBlue}{$\blacksquare$}}} 
[.{{\small\textcolor{Gray}{$\bigcirc$}}} ]
[.{{\small\textcolor{Gray}{$\bigcirc$}}} ]
[.{{\Large\textcolor{Gray}{$\square$}}} ]]]
[.{{\LARGE\textcolor{LightSteelBlue}{$\blacksquare$}}} 
[.{{\LARGE\textcolor{LightSteelBlue}{$\blacksquare$}}} 
[.{{\small\textcolor{Gray}{$\bigcirc$}}} ]
[.{{\small\textcolor{Gray}{$\bigcirc$}}} ]
[.{{\Large\textcolor{Gray}{$\square$}}} ]]
[.{{\LARGE\textcolor{LightSteelBlue}{$\blacksquare$}}} 
[.{{\small\textcolor{Gray}{$\bigcirc$}}} ]
[.{\textcolor{Red}{$\CIRCLE$}} ]
[.{{\small\textcolor{Gray}{$\bigcirc$}}} ]]
[.{{\Large\textcolor{Gray}{$\square$}}} ]]
[.{{\Large\textcolor{Gray}{$\square$}}} ]]
\end{tikzpicture}
\end{minipage}

\vspace{1mm}
\begin{minipage}[t]{0.39\linewidth}\vspace{0pt}
\begin{tikzpicture}[scale=0.45]
\draw [line width=1.6952] (0,0)--(8,0);
\draw [line width=1.6952] (0,0)--(0,7);
\draw [line width=1.6952] (8,0)--(8,7);
\draw [line width=0.3,black] (1,0)--(1,7);
\draw [line width=0.3,black] (2,0)--(2,7);
\draw [line width=0.3,black] (3,0)--(3,7);
\draw [line width=0.3,black] (4,0)--(4,7);
\draw [line width=0.3,black] (5,0)--(5,7);
\draw [line width=0.3,black] (6,0)--(6,7);
\draw [line width=0.3,black] (7,0)--(7,7);
\draw [line width=0.62,black] (2,0)--(2,7);
\draw [line width=0.62,black] (4,0)--(4,7);
\draw [line width=0.62,black] (6,0)--(6,7);
\draw [line width=1.068,black] (4,0)--(4,7);
\draw [line width=0.3,black] (0,1)--(1,1);
\draw [semithick, black, fill=red] (0.5,0.5) circle [radius=0.25];
\draw [semithick, gray] (0.5,1.5) circle [radius=0.25];
\draw [line width=0.62,black] (1,3)--(2,3);
\draw [semithick, black, fill=red] (1.5,2.5) circle [radius=0.25];
\draw [semithick, gray] (1.5,3.5) circle [radius=0.25];
\draw [line width=0.62,black] (0,3)--(1,3);
\draw [semithick, gray] (0.5,3.5) circle [radius=0.25];
\draw [line width=1.068,black] (2,5)--(3,5);
\draw [semithick, gray] (2.5,5.5) circle [radius=0.25];
\draw [line width=0.3,black] (3,2)--(4,2);
\draw [semithick, black, fill=red] (3.5,1.5) circle [radius=0.25];
\draw [semithick, gray] (3.5,2.5) circle [radius=0.25];
\draw [line width=1.068,black] (3,5)--(4,5);
\draw [semithick, black, fill=red] (3.5,4.5) circle [radius=0.25];
\draw [semithick, gray] (3.5,5.5) circle [radius=0.25];
\draw [line width=1.068,black] (0,5)--(1,5);
\draw [semithick, gray] (0.5,5.5) circle [radius=0.25];
\draw [line width=1.068,black] (1,5)--(2,5);
\draw [semithick, gray] (1.5,5.5) circle [radius=0.25];
\draw [line width=0.3,black] (7,4)--(8,4);
\draw [semithick, black, fill=red] (7.5,3.5) circle [radius=0.25];
\draw [semithick, gray] (7.5,4.5) circle [radius=0.25];
\end{tikzpicture}
\end{minipage}   
\vspace{1mm}
\begin{minipage}[t]{0.59\linewidth}\vspace{0pt}
\begin{tikzpicture}[scale=0.6]
\Tree
[.{{\LARGE\textcolor{LightSteelBlue}{$\blacksquare$}}} 
[.{{\LARGE\textcolor{LightSteelBlue}{$\blacksquare$}}} 
[.{{\LARGE\textcolor{LightSteelBlue}{$\blacksquare$}}} 
[.{\textcolor{Red}{$\CIRCLE$}} ]
[.{\textcolor{Red}{$\CIRCLE$}} ]
[.{{\small\textcolor{Gray}{$\bigcirc$}}} ]]
[.{{\LARGE\textcolor{LightSteelBlue}{$\blacksquare$}}} 
[.{{\small\textcolor{Gray}{$\bigcirc$}}} ]
[.{\textcolor{Red}{$\CIRCLE$}} ]
[.{\textcolor{Red}{$\CIRCLE$}} ]]
[.{{\LARGE\textcolor{LightSteelBlue}{$\blacksquare$}}} 
[.{{\small\textcolor{Gray}{$\bigcirc$}}} ]
[.{{\small\textcolor{Gray}{$\bigcirc$}}} ]
[.{{\Large\textcolor{Gray}{$\square$}}} ]]]
[.{{\LARGE\textcolor{LightSteelBlue}{$\blacksquare$}}} 
[.{{\LARGE\textcolor{LightSteelBlue}{$\blacksquare$}}} 
[.{{\small\textcolor{Gray}{$\bigcirc$}}} ]
[.{{\small\textcolor{Gray}{$\bigcirc$}}} ]
[.{{\Large\textcolor{Gray}{$\square$}}} ]]
[.{{\LARGE\textcolor{LightSteelBlue}{$\blacksquare$}}} 
[.{{\small\textcolor{Gray}{$\bigcirc$}}} ]
[.{\textcolor{Red}{$\CIRCLE$}} ]
[.{{\small\textcolor{Gray}{$\bigcirc$}}} ]]
[.{{\Large\textcolor{Gray}{$\square$}}} ]]
[.{{\LARGE\textcolor{LightSteelBlue}{$\blacksquare$}}} 
[.{{\LARGE\textcolor{LightSteelBlue}{$\blacksquare$}}} 
[.{{\small\textcolor{Gray}{$\bigcirc$}}} ]
[.{{\small\textcolor{Gray}{$\bigcirc$}}} ]
[.{{\Large\textcolor{Gray}{$\square$}}} ]]
[.{{\LARGE\textcolor{LightSteelBlue}{$\blacksquare$}}} 
[.{{\small\textcolor{Gray}{$\bigcirc$}}} ]
[.{{\small\textcolor{Gray}{$\bigcirc$}}} ]
[.{{\Large\textcolor{Gray}{$\square$}}} ]]
[.{{\Large\textcolor{Gray}{$\square$}}} ]]]
\end{tikzpicture}
\end{minipage}
\caption{The integer grid and the resulting tree as a result of executing a sequence of \kwwrite{} operations.}

\label{fig:insseq}

\end{figure}
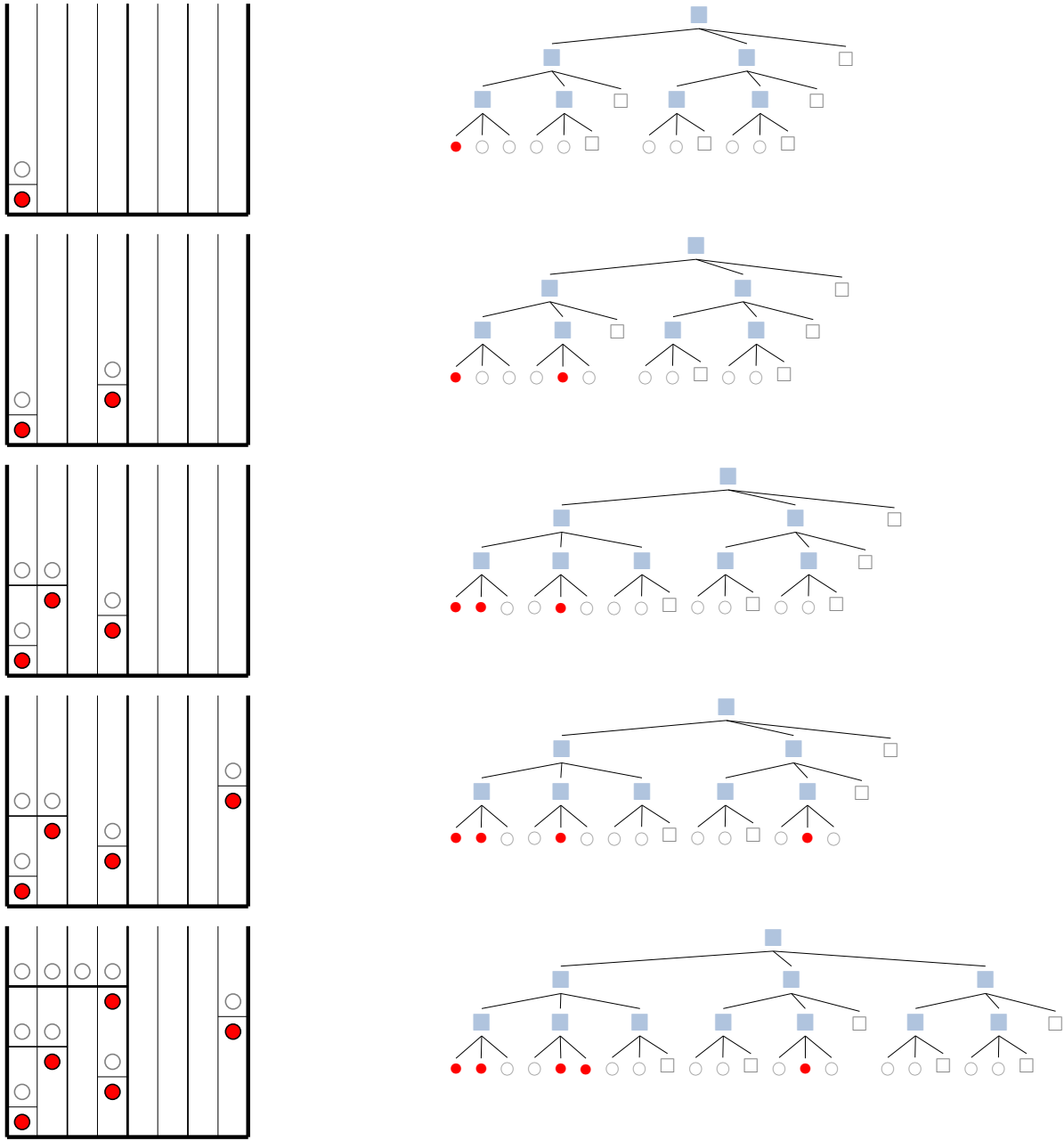

%% file: ozanalysis2.tex
\section{Analysis}
\label{sec:analysis}

\subsection{Space usage}

\begin{lem} 
\label{lem:nontopspace}
The space usage of a bottom tree is \bigoh{U \log U}. 
\end{lem}

\begin{proof}
Since each internal node in the tree has at least two children, we only need to bound only the number of leaves in the tree. 
The leaves of the tree can be partitioned into two sets depending of whether or not they are full. 
Recall that a leaf is full if and only if its rectangle intersects a point in $P$. If a leaf is not full, we call it an empty leaf.
Note that each empty leaf contains the answer to the \kwpread{} operation at the bottom of its rectangle. 
There are $U$ points from $P$ in the square of each bottom tree and we will charge the empty leaves to these $U$ points.
Observe that an empty leaf is only created when we create the initial tree or when we create a new complete subtree as the third child of a node in the tree. 
In the former case, there are at most $U$ leaves and this happens once per each tree. 
In the latter case, call the third child $q$ and denote by $p$ the sibling of $q$ that has a rectangle below the rectangle of $q$. 
Note that by construction, $p$ is full. 
Assuming $q$ and $p$ are of height $h$, by Lemma~\ref{lem:rectlb}, there are at least $2^h$ points from $P$ in the rectangle of $p$. We can charge the $2^h$ leaves created in the subtree of $q$ to these points. How many times can such a point be charged? It is charged at most once per level. There are $\log U$ levels. Thus, there are at most \bigoh{U\log U} leaves in a bottom tree. Since each bottom tree is stored contiguously in space linear in the number of leaves, the lemma follows. 
\end{proof}

\begin{lem}
\label{lem:spacecomplexity}
The space usage of the entire structure is \bigoh{U^{\log 3}+V \log U}.
\end{lem}
\begin{proof}
The top tree uses space \bigoh{U^{\log 3}}. From Lemma~\ref{lem:nontopspace}, the \bigoh{V/U} bottom ST-trees use space \bigoh{U \log U} each. The log of \kwwrite{} operations is size \bigoh{V}. The current memory array is size \bigoh{U}. Other global information takes size \bigoh{1}.
\end{proof}

\subsection{Comments on memory allocation}
According to the ideal cache assumption, in the cache-oblivious model the runtime is within a constant factor of that which has the optimal movement of blocks in and out of cache. Thus, we can assume a particular active memory management strategy and a runtime of using this strategy is an upper bound of the runtime in the cache-oblivious model. In particular, we can base our assumed memory management strategy on the one an ephemeral structure would use on the same sequence of operations, with particular memory and block sizes.

Globally, we assume the memory $M$ is split into three equal parts, with each part dedicated to the execution of each of the three primitive operations.

\subsection{Analysis of Read}

\begin{thm}
Suppose a sequence $X$ of \kwread{} operations takes time $T(M,B)$ to execute 
ephemerally in the cache-oblivious model on a machine with memory $M$ and block size $B$ and an initially empty memory. Then in our structure, the runtime is 
\bigoh{T(M/3,B)}.
\end{thm}

\begin{proof}
Since executing \kwread{} operations is simply done using the current memory buffer, and there is a memory of size $M/3$ allocated for \kwread{}s, this gives the result.
\end{proof}

\subsection{Analysis of Persistent-Read}

\begin{lem}
\label{lem:blowup}
Consider the leaves of the ST-tree corresponding to a subarray of ephemeral memory of size $w$ at a fixed time $v$ ($A_{v}[i\ldots i+w-1]$ for some $i$). These leaves and all of their ancestors in the ST-tree are contained in \varbigoh{\frac{w}{B^{\miniexpin}}+\log_BU} blocks
\end{lem}
\begin{proof}

Let $L$ be the set of all points corresponding to the memory locations $A_{v}[i\ldots i+w-1]$.
There are two disjoint rectangles corresponding to nodes in the ST-tree with width at most $2w$ such that the two rectangles contain and partition all elements of $L$. Let $x_l$ and $x_r$ be the nodes corresponding to these two rectangles. Those leaves in subtrees of $x_l$ and $x_r$ corresponding to points in $L$ are the leaves of two binary trees with roots $x_l$ and $x_r$.
 Since the height of $x_l$ and $x_r$ are at most $1+\log w$, by Lemma~\ref{lem:binarysubgraph} any binary tree rooted at them will occupy at most $\varbigoh{1+\frac{w}{B^{\miniexpin}}}$ blocks.
The nodes on the path between $x_{l}$ and $x_{r}$ are stored in \bigoh{\log_{B} U} blocks.
\end{proof}

\begin{thm}
Suppose a sequence $X$ of \kwpread{} operations executed on the same version takes time $T(M,B)$ to execute 
ephemerally in the cache-oblivious model on a machine with memory $M$ and block size $B$ and an initially empty memory. Then in our structure, the runtime is 
$$\varbigoh{ T\left(\frac{\frac 13 M}{B^{1-\miniexp}(1+\log_BU)}, B^{\miniexp} \right)(1+\log_BU) }$$ 
\end{thm}

\begin{proof}
Consider executing $X$ ephemerally with memory $\frac{M}{3(1+\log_BU)B^{1-\miniexp}}$ and block size $B^{\miniexp} $. It keeps in memory 
\inout{
\frac{M}{3(1+\log_BU)B^{1-\miniexp}}\cdot \frac{1}{B^{\miniexp}}  
=
\frac{M}{3(1+\log_BU)B}
}
ephemeral blocks.
Each ephemeral block is stored in $\varbigoh{B^{\miniexp}/B^{\miniexp}}=\bigoh{1}$ blocks; including the ancestors of the block this becomes $\bigoh{1+\log_B U}$.
Now,  by Lemma~\ref{lem:blowup} the memory blocks needed to keep the leaves representing $\frac{M}{3(1+\log_BU)B}$ consecutive memory locations of length $B$ in the persistent structure and their ancestors is 
\inout{\frac{M}{3(1+\log_BU)B}\cdot 
(1+\log_BU) = \frac{M}{3B}.
}
Thus they can all be stored in the third of memory allocated to \kwpread{} operations ($\frac{M}{3B}$ blocks), and thus moving $p$ to any location in the memory in the ephemeral structure will involve walking it entirely though locations in memory in the persistent structure, and thus will have no cost.
Now suppose the ephemeral structure moves one block into memory at unit cost. This could require moving $\varbigoh{B^{\miniexp}/B^{\miniexp}+\log_BU}= \bigoh{1+\log_BU}$ blocks to move the associated leaves and their ancestors in the persistent structure into memory by Lemma~\ref{lem:blowup}. Thus, this is the slowdown factor the persistent structure will incur relative to the ephemeral.
\end{proof}

\subsection{Analysis of Write}

We charge the cost of rebuilding to \kwwrite{} operations. This only increases their cost by a multiplicative constant factor since we double $\asize$ after \bigom{\asize} \kwwrite{}s. 

Every $U$ \kwwrite{}s, we make the existing top tree a bottom tree and create a new top tree. Three steps involved in this process are 
closing the rectangles of the top tree, 
compression of the top tree to a bottom tree, 
and initializing a new top tree and populating the leafs.

The following lemmas bound the cost of these steps. 

\begin{lem}
\label{lem:copyclose}
Copying performed to populate the leafs of a newly created subtree of height $h$, or closing the open rectangles in the subtree of a node at height $h$ costs \varbigoh{\frac{2^{h}}{B^{\miniexpin}}} block transfers. 
\end{lem}
\begin{proof}
We only close rectangles that are open and we copy nodes from the array $C$ to leafs of the new subtree which only has open rectangles. 
If a node is closed all of its children are closed and do not need to be traversed. Given any node, since at most 2 out of its potentially 3 children can be open, the tree we traverse to close rectangles or copy nodes is a binary tree. By Lemma~\ref{lem:binarysubgraph}, each of these tasks costs \varbigoh{\frac{2^{h}}{B^{\miniexpin}}} block transfers. 
\end{proof}

\begin{lem}
\label{lem:compress}
Compression of the top tree into a bottom tree costs \varbigoh{\frac{U\log U}{B^{\miniexpin}}} block transfers.
\end{lem}
\begin{proof} 
The top tree is initially a complete binary tree and is stored in \varbigoh{\frac{U}{B^{\miniexpin}}} blocks by Lemma~\ref{lem:binarysubgraph}. We have to account for the additional blocks used to store the elements inserted and nodes created in the top tree after it was created.

By Lemma~\ref{lem:rebuild}, \bigoh{\frac{U}{2^{k}}} nodes are added to the top tree at height $k$. 
When a node at height $k$ is added to the top tree, it is the root of a complete binary search tree
and its subtree is stored in \varbigoh{2^{k}/B^{\miniexpin}} blocks by Lemma~\ref{lem:binarysubgraph}. 
This implies that added nodes take an additional 
\inout{
\sum_{j=1}^{\log U} \frac{U}{2^{j}}\cdot \frac{2^{j}}{B^{\miniexpin}} = \varbigoh{\frac{U\log U}{B^{\miniexpin}}}
} 
blocks. 
\end{proof}

\begin{thm}
\label{thm:k}
Suppose a sequence of $k$ \kwwrite{} operations takes time $T(M,B)$ to execute 
ephemerally in the cache-oblivious model on a machine with memory $M$ and block size $B$ and an initially empty memory. Then in our structure, the runtime is 
\[
\varbigoh{ T\left(\frac{\frac13 M}{B^{1-\miniexp}}, B^{\miniexp} \right)  + k\frac{\log U}{B^{\miniexp }}}
\]
\end{thm}
\begin{proof}
To find the item, the analysis is similar to that of \kwpread{}, except we can use the pointers in $C$ to directly go to the item. This removes the $\log_B U$ terms. 

We need to bound the cost of operations performed to maintain the ST-tree. Recall that after each insertion, we go up the tree until we find an ancestor $p$ such that its parent is not full. 
Then, we create a new sibling node $q$ whose leafs are populated with the copies of values stored in the corresponding top leaf nodes of $p$, and close the rectangles in the subtree of $p$. 
Letting the height of $p$'s parent be $h$, we refer to this sequence of operations as the expansion of a node at height $h$.

The creation of the new top tree occurs once every $U$ \kwwrite{}s, and thus the amortized cost per \kwwrite{} is \bigoh{\log U/B^{\miniexpin}} by Lemma~\ref{lem:compress}. 
Since the expansion of a node at height $h$ happens because there have been at least $2^{h}$ insertions since the node was created, by Lemma~\ref{lem:copyclose}, the amortized cost of a \kwwrite{} operation is at most
\inout{
\frac{\log U}{B^{\miniexp}}+ \sum_{j=1}^{\log U} \frac{1}{2^j} \cdot\frac{2^{j}}{B^{\miniexp}}
= \varbigoh{\frac{\log U}{B^{\miniexp }}}.
}
\end{proof}

We also consider the case when all \kwwrite{} operations are executed in unique memory cells. 
This is a reasonable assumption when we have update operations involving numerous  \kwwrite{} operations where for each memory cell we only need to remember the last \kwwrite{} executed during that update operation.

\begin{thm}
\label{thm:nok}
Suppose a sequence of \kwwrite{} operations performed on unique cells of $A$ takes time $T(M,B)$ to execute 
ephemerally in the cache-oblivious model on a machine with memory $M$ and block size $B$ and an initially empty memory. Then in our structure, the runtime is 
\[
\varbigoh{
T\left(
\frac{\frac13 M}{B^{1-\miniexp}\log B}, 
\frac{B^{\miniexp}}{\log B} 
\right)
\log_B U 
}
\]
\end{thm}

\begin{proof}
Observe that given $M',M'',B',B''$ if $B'\leq B''$ and $M'/B' = M''/B''$, then $T(M',B') \geq T(M'',B'')$. 
Thus, since no two \kwwrite{} operations overlap, we have in the worst case that
\inout{
k
\leq  T\left(
\frac{\frac13 M}{B^{1-\miniexp}\log B}, 
\frac{B^{\miniexp}}{\log B} 
\right)
\cdot  
\frac{B^{\miniexp}}{\log B}.
} 
The result then follows by substituting $k$ in Theorem~\ref{thm:k}.
\end{proof}